\documentclass[reqno, 12pt, a4paper]{amsart}

\usepackage{amsmath}
\usepackage{amsthm}
\usepackage{amssymb}
\usepackage{amsfonts}
\usepackage{mathtools}

\usepackage{fullpage}

\usepackage{dsfont}

\usepackage{mathrsfs}

\usepackage[dvipsnames,svgnames]{xcolor}
\colorlet{MyBlue}{DodgerBlue!60!Black}
\colorlet{MyGreen}{DarkGreen!85!Black}

\definecolor{ngreen}{RGB}{56, 188, 83}
\definecolor{nred}{RGB}{196, 39, 39}

\usepackage{subfigure}
\usepackage{tikz}
\usetikzlibrary{calc,patterns}
\usetikzlibrary{arrows}
\usetikzlibrary{arrows.meta}
\usetikzlibrary{decorations.markings}
\usetikzlibrary{shapes.misc}
\usepackage{diagbox}
\usepackage{longtable}
\usepackage{wrapfig}

\usepackage{acronym}
\usepackage{booktabs}       
\usepackage{ifthen}
\usepackage{latexsym}
\usepackage{nicefrac}       
\usepackage{wasysym}
\usepackage{xspace}
\usepackage[shortlabels]{enumitem}

\tikzset{cross/.style={cross out, draw=black, minimum size=2*(#1-\pgflinewidth), inner sep=0pt, outer sep=0pt},cross/.default={1pt}}

\usepackage[authoryear,comma]{natbib}

\usepackage{hyperref}
\hypersetup{
colorlinks=true,
linktocpage=true,
pdfstartview=FitH,
breaklinks=true,
pdfpagemode=UseNone,
pageanchor=true,
pdfpagemode=UseOutlines,
plainpages=false,
bookmarksnumbered,
bookmarksopen=false,
bookmarksopenlevel=1,
hypertexnames=true,
pdfhighlight=/O,
urlcolor=MyBlue!60!black,linkcolor=MyBlue!70!black,citecolor=DarkGreen!70!black, 
pdftitle={},
pdfauthor={},
pdfsubject={},
pdfkeywords={},
pdfcreator={pdfLaTeX},
pdfproducer={LaTeX with hyperref}
}

\numberwithin{equation}{section}  
\usepackage[sort&compress,capitalize,nameinlink]{cleveref}
\crefname{app}{Appendix}{Appendices}

\crefrangeformat{equation}{\upshape(#3#1#4)\textendash(#5#2#6)}

\usepackage[textwidth=30mm]{todonotes}

\newcommand{\debug}[1]{{\color{purple}#1}}

\usepackage[xcolor]{changes}
\definechangesauthor[name=Ben, color=blue]{Ben}
\definechangesauthor[name=Andrea, color=red]{Andrea}
\definechangesauthor[name=Marco, color=magenta]{Marco}
\definechangesauthor[name=Ziwen, color=orange]{Ziwen}

\theoremstyle{plain}
\newtheorem{theorem}{Theorem}

\newtheorem*{corollary*}{Corollary}
\newtheorem{lemma}[theorem]{Lemma}
\newtheorem{proposition}[theorem]{Proposition}

\theoremstyle{definition}
\newtheorem{definition}[theorem]{Definition}
\newtheorem*{definition*}{Definition}

\newtheorem*{hypothesis*}{Hypothesis}

\theoremstyle{remark}

\newtheorem*{remark*}{Remark}
\newtheorem*{notation*}{Notational remark}

\numberwithin{theorem}{section}
\numberwithin{remark}{section}
\numberwithin{example}{section}

\newcommand{\N}{\mathbb{N}}

\newcommand{\PNE}{\mathcal{E}}

\newcommand{\card}[1]{\mathsf{\debug{card}}\parens*{#1}}

\DeclareMathOperator{\Poisson}{\mathsf{\debug{Poisson}}}

\DeclarePairedDelimiter{\braces}{\{}{\}}
\DeclarePairedDelimiter{\bracks}{[}{]}
\DeclarePairedDelimiter{\parens}{(}{)}

\DeclarePairedDelimiter{\floor}{\lfloor}{\rfloor}

\DeclarePairedDelimiterX{\braket}[2]{\langle}{\rangle}{#1,#2}

\DeclarePairedDelimiterX{\inner}[2]{\langle}{\rangle}{#1,#2}
\DeclarePairedDelimiterX{\setdef}[2]{\{}{\}}{#1:#2}

\DeclarePairedDelimiterXPP{\probof}[1]{\prob}{(}{)}{}{%

#1}

\DeclarePairedDelimiterXPP{\exof}[1]{\ex}{[}{]}{}{%

#1}

\newacro{BRD}{best-response dynamics}
\newacro{brd}[bRD]{better-response dynamics}
\newacro{CLT}{central limit theorem}
\newacro{NE}{Nash equilibrium}
\newacroplural{NE}[PE]{Nash equilibria}
\newacro{PNE}{pure Nash equilibrium}
\newacroplural{PNE}[PNE]{pure Nash equilibria}
\newacro{SPNE}{strict pure Nash equilibrium}
\newacroplural{SPNE}[SPNE]{strict pure Nash equilibria}
\newacro{WHP}{with high probability}
\newacro{WVHP}{with very high probability}

\def\cnst{\text{cnst}\,~}

\title{When ``Better'' is better than ``Best''}

\author{Ben Amiet}
\author{Andrea Collevecchio}
\author{Kais Hamza}
\address{School of Mathematics, Monash  University, Melbourne, Australia}
\email{ben.amiet@monash.edu,andrea.collevecchio@monash.edu,kais.hamza@monash.edu}



\begin{document}

\maketitle

\vspace{-.7cm}
\begin{abstract}
    We consider two-player normal form games where each player has the same finite strategy set.
    The payoffs of each player are assumed to be i.i.d.\ random variables with a continuous distribution. 
    We show that, with high probability, the better-response dynamics converge to pure Nash equilibrium whenever there is one, whereas best-response dynamics fails to converge, as it is trapped.
\end{abstract}


\section{Introduction}


\subsection{Background and motivation}
Among the various techniques to find a \ac{PNE} in a normal form game, \ac{BRD} is one of the simplest techniques to describe. 
Starting from some strategy profile, one player, picked at random, chooses the strategy which guarantees them the highest payoff, given the other players' chosen strategies.
The procedure is then repeated, starting from the new profile. 
When no player can choose a strategy that improves their payoff, a \ac{PNE} is reached.
The procedure clearly does not converge in games that do not have a \acp{PNE}. 
Existence of \acp{PNE} does not guarantee convergence however, as the procedure could also fail by cycling indefinitely on a set of profiles. 

An alternative procedure, called \ac{brd}, requires the player who is chosen at random to move to a new strategy that guarantees a higher (but not necessarily maximal) payoff than the present one.
At first sight, it seems self-evident that \ac{BRD} should perform better than \ac{brd}.
To be a {PNE}, a profile must maximize the payoff of each player, given the other players' strategies; therefore, choosing a better response that is not a best response is sub-optimal.
In this paper, we prove that this intuition is in general false.
We consider two-player normal form games where each player has the same finite strategy set.
The payoffs of each player are assumed to be i.i.d.\ random variables with a continuous distribution. 
We show that, with high probability, bRD converges to a PNE whenever one exists, whereas BRD will cycle indefinitely on a subset of strategy profiles.


\subsection{Related work}

The literature about the number of \acp{PNE} in games with random payoffs is quite extensive. 
We refer the reader to \citet{AmiColScaZho:MORfrth} for a list of the main papers on the topic.
In particular \citet{Pow:IJGT1990} proved that, in a normal form game where payoffs are i.i.d.\ with a continuous distribution, when the number of strategies of at least two players goes to infinity, the distribution of the number of \acp{PNE} converges to $\Poisson(1)$.
See also \citet{RinSca:GEB2000}.

The use of  \ac{BRD} to find \acp{PNE} has been studied, among others, by \citet{Blu:GEB1993,You:E1993,FriMez:JET2001,TakYam:EB2002} and \citet{FabJagSch:TCS2013}.
In general, \ac{BRD} does not converge to a \ac{PNE}. 
It is known to converge,  for instance, in potential games, as defined by \citet{MonSha:GEB1996}.
The performance of \ac{BRD} in potential games with a random potential function has been studied in \citet{CouDurGauTou:NetGCoop2014,DurGau:AGT2016} and \citet{DurGarGau:PE2019}.

\citet{GoeMirVet:FOCS2005} defined the concept of \emph{sink equilibria}  (which in this paper are called traps), i.e., sets of strategy profiles where a \ac{BRD} may end up cycling.
\citet{ChrMirSid:TCS2012} studied  the rate of convergence of a \ac{BRD} to approximate solutions of a game;
\citet{DutKes:SODA2017} considered \ac{BRD} in the context of combinatorial auctions.

The \ac{brd} process is a  less extensively studied concept: 
\citet{Ren:E1999,Ren:ET2011,Kuk:JME2018} studied its behavior in infinite games with discontinuous payoffs;
\citet{CabSer:GEB2011} studied the dynamics in the framework of implementation;
\citet{FabJagSch:TCS2013} studied it in the context of weakly acyclic games and show the relation between weak acyclicity and existence of \acp{PNE}.


\section{Notation and Main results}
\label{se:notation}

Given integers $k\le K$, we set $\bracks{k,K} \coloneqq\braces*{k,\dots,K}$ and $\bracks{K} \coloneqq \bracks{1,K}$. 
We consider a two-player game where both players can choose a strategy in $\bracks{K}$; let $\mathcal{S} \coloneqq \bracks{K}^2$.
For a strategy profile $\boldsymbol{s}\coloneqq\parens*{s_{1},s_{2}}\in \mathcal{S}$, the payoff of player $i$ is $Z_{i}^{\boldsymbol{s}}=Z_{i}^{s_{1}, s_{2}}$.
We will assume that the payoffs $Z_{i}^{\boldsymbol{s}}$ are i.i.d. and follow a continuous distribution.

Two strategy profiles $\boldsymbol{s} = (s_1, s_2)$ and $\boldsymbol{t}=(t_1, t_2)$ are neighbors, denoted by $\boldsymbol{s}\sim \boldsymbol{t}$, if $\boldsymbol{s} \neq \boldsymbol{t}$ and  $s_i=t_i$ for exactly one $i \in \{1,2\}$. We use  $\boldsymbol{s}\sim_i \boldsymbol{t}$ to denote that $\boldsymbol{s}\sim \boldsymbol{t}$ and that these strategy profiles have different strategies for player $i$.
A strategy profile $\boldsymbol{s}$ is a \emph{pure Nash equilibrium}  (\ac{PNE}) if $Z_i^{\boldsymbol{s}} \ge Z_i^{\boldsymbol{t}}$ for any $\boldsymbol{s}\sim_i \boldsymbol{t}$, for  $i\in \{1,2\}$.
For $i\in \{1,2\}$ define the sets
$$\mathcal{P}^{(i)}_{\boldsymbol{s}} \coloneqq \{\boldsymbol{t} \in \mathcal{S}\colon \boldsymbol{t} \sim_{i} \boldsymbol{s}, Z_i^{\boldsymbol{t}} \geq Z_i^{\boldsymbol{s}}\} \text{ and } \mathcal{M}^{(i)}_{\boldsymbol{s}} \coloneqq \{\boldsymbol{t} \in \mathcal{P}^{(i)}_{\boldsymbol{s}}\colon Z_i^{\boldsymbol{t}} \geq Z_i^{\boldsymbol{u}} \text{ for all } \boldsymbol{u} \sim_i \boldsymbol{s}\}.$$

\begin{definition}
    \label{de:BRD}
    The \acfi{brd}\acused{brd}, denoted $\mathsf{bRD} = \{\mathsf{bRD}(n)\}_{n=0}^\infty$ with $\mathsf{bRD}(0) = (1,1)$, is a discrete-time process on $\mathcal{S}$ that evolves as follows.
    At time $n+1$, pick a player at random, independently of the current value of the process and its past. 
    Call this random variable $I \in \{1, 2\}$, and choose $\mathsf{bRD}(n+1)$ uniformly at random from the set $\mathcal{P}^{(I)}_{\mathsf{bRD}(n)}$.
    If the latter set is empty, repeat the procedure with the other player.
    If $\mathcal{P}^{(1)}_{\mathsf{bRD}(n)} = \mathcal{P}^{(2)}_{\mathsf{bRD}(n)} = \varnothing$ then set $\mathsf{bRD}(n+1) = \mathsf{bRD}(n)$; at this point, $\mathsf{bRD}$ has reached a \ac{PNE}, and we say that the process has converged.
    
    The \acfi{BRD}\acused{BRD}, denoted $\mathsf{BRD} = \{\mathsf{BRD}(n)\}_{n=0}^\infty$ with $\mathsf{BRD}(0) = (1,1)$, is defined similarly, except it chooses its strategy profile at time $n+1$ from the set $\mathcal{M}^{(I)}_{\mathsf{BRD}(n)}$. Its convergence criterion follows mutatis mutandis.
\end{definition}


Note that because $\mathsf{BRD}$ always moves to a best response for its current strategy profile and because the payoff distribution is continuous, $\mathsf{BRD}$ will always alternate between changing the action of player 1 and 2 until it reaches a PNE (if it does at all).
Both \ac{BRD} and \ac{brd} clearly fail to converge in games that do not admit \acp{PNE}, but they could also keep cycling in games that do have \acp{PNE}. 
Below we define the structures upon which these processes indefinitely cycle.

\begin{definition}
    \label{de:trap}
    Let $\mathsf{X}$ denote either of the two processes $\mathsf{BRD}$ or $\mathsf{bRD}$. A nonempty subset of strategy profiles $\tau\subset \mathcal{S}$ such that $\card{\tau} \geq 2$ is called an \emph{$\mathsf{X}$-trap} if $\{\mathsf{X}(n)\in\tau\} \subset \{\mathsf{X}(n+1)\in\tau\}$ and, for all $\boldsymbol{s}\in\tau$, 
    \begin{equation*}
        \{\mathsf{X}(n)\in\tau\} \subset \{\inf\{k \in (n, \infty)\cap \N \colon \mathsf{X}(k)=\boldsymbol{s}\} < \infty \mbox{ a.s.}\}
    \end{equation*}
or equivalently if $\tau\subset\liminf_k\{\mathsf{X}(k)\}$ and $\mathsf{X}(n+1)\in\tau$ whenever $\mathsf{X}(n)\in\tau$.
\end{definition}

We note that for either $\mathsf{BRD}$ or $\mathsf{bRD}$ to not converge, it must enter a trap. 
Also, by restricting the size of a trap to be at least 2, \ac{PNE} are excluded from this definition.
An easy reasoning will convince the reader that both types of trap must contain at least 4 strategy profiles.
The question that we want to address is with what likelihood do both \ac{BRD} and \ac{brd} converge in a game.
Denote by $\PNE_K$ the collection of PNE in the random game. 

\begin{theorem}\label{th:main1} 
\
    \begin{enumerate} 
        \item 
       $
            \label{eq:P-BRD-not-converge-no-atom}
           \lim_{K \to \infty} \mathbb{P}(\mathsf{BRD} \text{ does not converge}) =1.
       $
        \item 
        $
            \label{eq:P-BRD-not-converge-no-atom1}
            \mathbb{P}(\mathsf{bRD} \text{ converges}) = 1- {\rm e}^{-1} +O\left( 1/K\right).
        $ This implies that  $\mathbb{P}(\mathsf{bRD}$  converges$\;|\; \PNE_K \neq \varnothing  ) = 1 + o(1)$.
    \end{enumerate}
\end{theorem}

\section{Best Response Dynamics: Proof of \cref{th:main1} (1)}
\label{se:dynamics}

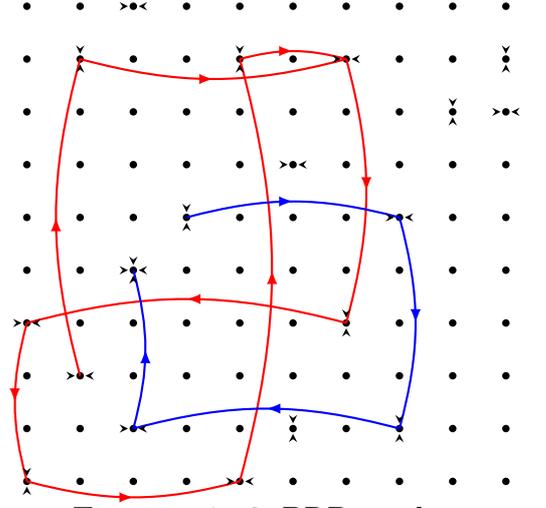
\begin{wrapfigure}{r}{0.42\textwidth}
    \centering
    \begin{tikzpicture}[scale=0.7]
        \foreach \x in {1,...,10}
        \foreach \y in {1,...,10}{
            \fill (\x,\y) circle (2pt);
        }
    	\begin{scope}[thick, red, decoration={markings, mark=at position 0.5 with {\arrow{latex}}}]
    	    \draw[postaction={decorate}] (2,3) to[out=105, in=-105] (2,9);
    	    \draw[postaction={decorate}] (2,9) to[out=-15, in=-165] (7,9);
    	    \draw[postaction={decorate}] (7,9) to[out=-75, in=75] (7,4);
    	    \draw[postaction={decorate}] (7,4) to[out=165, in=15] (1,4);
    	    \draw[postaction={decorate}] (1,4) to[out=-105, in=105] (1,1);
    	    \draw[postaction={decorate}] (1,1) to[out=-15, in=-165] (5,1);
    	    \draw[postaction={decorate}] (5,1) to[out=75, in=-75] (5,9);
    	    \draw[postaction={decorate}] (5,9) to[out=15, in=165] (7,9);
    	\end{scope}
    	\begin{scope}[thick, blue, decoration={markings, mark=at position 0.5 with {\arrow{latex}}}]
    	    \draw[postaction={decorate}] (4,6) to[out=15, in=165] (8,6);
    	    \draw[postaction={decorate}] (8,6) to[out=-75, in=75] (8,2);
    	    \draw[postaction={decorate}] (8,2) to[out=165, in=15] (3,2);
    	    \draw[postaction={decorate}] (3,2) to[out=75, in=-75] (3,5);
    	\end{scope}
        \foreach \x/\y in {5/1, 3/2, 2/3, 1/4, 3/5, 8/6, 6/7, 10/8, 7/9, 3/10}{
            \draw [->, >=stealth] (\x - 0.2,\y) -- (\x - 0.1, \y);
            \draw [->, >=stealth] (\x + 0.2,\y) -- (\x + 0.1, \y);
        }
    	\foreach \x/\y in {1/1, 2/9, 3/5, 4/6, 5/9, 6/2, 7/4, 8/2, 9/8, 10/9}{
            \draw [->, >=stealth] (\x,\y - 0.2) -- (\x, \y - 0.1);
            \draw [->, >=stealth] (\x,\y + 0.2) -- (\x, \y + 0.1);
        }
    \end{tikzpicture}
    \caption{2 BRD paths, one enters a trap (red), the other a PNE (blue)} \label{fi:BRD}
\end{wrapfigure}

For aesthetic purposes, in what follows, we implicitly condition on the event in which $\mathsf{BRD}(0)$ is a best response for exactly one player. If it is instead not a best response for either player, after one step it will arrive at a best response for the strategy which remained constant, and the process behaves as if it were under our implied condition from then on. Finally,  the probability that the starting vertex is a best response for both players -- i.e., a PNE -- is $1/K^2$, and does not influence our final result.

First we note that, because the payoffs follow a continuous distribution, there exists exactly one best response for each action either player chooses. It follows that $\mathsf{BRD}$ is trapped if and only if it visits a row or column it has visited previously. Moreover, $\mathsf{BRD}$ can only visit at most two strategy profiles along any given row or column: the strategy profile from which it enters a row or column; and the corresponding best response. As such, the maximum number of steps $\mathsf{BRD}$ can make before revisiting a row or column is at most $2K-2$. Hence, $\mathsf{BRD}$ is trapped if and only if it has not reached a PNE by this time.

Due to the fact that there is exactly one best response for each action, $\mathsf{BRD}$ must alternate moving along rows or columns at each step (see \cref{fi:BRD} for an example of this behaviour). It follows that, at time $t$, $\mathsf{BRD}$ must avoid $\floor{t/2}$ rows or columns. As each step is to a best response, each new strategy profile that $\mathsf{BRD}$ reaches is a PNE with probability $1/K$. Letting $\mathcal{R}(t)$ be the set of all strategy profiles in the rows and columns that $\mathsf{BRD}$ has visited by time $t$, we have
\begin{equation}
    \label{eq:recursive-pne}
    \mathbb{P}\parens*{\mathsf{BRD}(t+1)\in\PNE_K} = \mathbb{P}\parens*{\mathsf{BRD}(t)\in\PNE_K} + \mathbb{P}\parens*{\mathsf{BRD}(t)\notin\PNE_K\cup\mathcal{R}(t-1)}\frac{K-1-\floor{t/2}}{K-1}\frac{1}{K}.
\end{equation}
In order for $\mathsf{BRD}$ to not be in a trap or a PNE by time $t+1$, it must: not be in a trap or PNE at time $t$; avoid all previously visited rows and columns; and not step to a PNE. It follows that
\begin{equation}
    \label{eq:recursive-free}
    \mathbb{P}\parens*{\mathsf{BRD}(t+1)\notin\PNE_K\cup\mathcal{R}(t)} = \mathbb{P}\parens*{\mathsf{BRD}(t)\notin\PNE_K\cup\mathcal{R}(t-1)}\frac{K-1-\floor{t/2}}{K-1}\frac{K-1}{K}.
\end{equation}
Applying \cref{eq:recursive-pne,eq:recursive-free} repeatedly, we obtain
\begin{equation}
    \label{eq:p-BRD-cv}
    \begin{aligned}
        \mathbb{P}\parens*{\mathsf{BRD} \text{ converges}} &= \mathbb{P}\parens*{\mathsf{BRD}(2K-2)\in\PNE_K}\\
        &= \mathbb{P}\parens*{\mathsf{BRD}(1)\in\PNE_K} + \mathbb{P}\parens*{\mathsf{BRD}(1)\notin\PNE_K\cup\mathcal{R}(0)}\parens*{\frac{1}{K-1}\sum_{t=1}^{2K-3}\prod_{j=1}^t \frac{K-1-\floor{t/2}}{K}}\\
        &= \frac{1}{K} + \frac{1}{K}\sum_{t=1}^{2K-3}\prod_{j=1}^t \frac{K-1-\floor{t/2}}{K}.
    \end{aligned}
\end{equation}

To bound the sum in \cref{eq:p-BRD-cv}, we make use of the inequality $e^x \geq 1 + x$, giving
\begin{equation*}
    \sum_{t=1}^{2K-3} \prod_{j=1}^t \frac{K - 1 - \lfloor j/2\rfloor}{K} \leq \sum_{t=1}^{2K-3} \exp\left(-\frac{1}{K}\sum_{j=3}^{t+2}\lfloor j/2\rfloor\right) \leq 2\sum_{t=1}^{K-2} \exp\left(-\frac{t^2}{K}\right).
\end{equation*}
Finally, we can bound this quantity above by making use of a Riemann sum:
\begin{equation}
    \label{eq:riemann-bd}
    \sqrt{K}\sum_{t=1}^{K-2} \frac{1}{\sqrt{K}}\exp\left(-\frac{t^2}{K}\right) \leq \int_0^\infty e^{-x^2} dx = \frac{\sqrt{K\pi}}{2}.
\end{equation}
This in turn yields
$\displaystyle\mathbb{P}(\mathsf{BRD} \text{ converges}) \leq \frac{1}{K} + \sqrt{\frac{\pi}{K}}.$

\section{Combinatorial bounds}
The proof of \cref{th:main1} (2) relies on a few combinatorial results which we include in this preliminary section. 
For the vector $\boldsymbol{c}=\parens{c_{1},\dots,c_{K}}$, we set $\ell(\boldsymbol{c}) = \sum_ {s=1}^{K}c_{s}$. 
For intuitive purposes, consider $\boldsymbol{c}$ to be the vector describing the number of strategy profiles in each column belonging to a given trap. 

\begin{proposition}\label{pr:comb}
    Let $m\in[K]$. For fixed $\boldsymbol{c}\in[0,m]^K$,
    \begin{equation}
    \prod_{i=1}^K c_i! (K-c_i)! \leq (m!)^{\floor{\ell(\boldsymbol{c})/m}}(K-m)!^{\floor{\ell(\boldsymbol{c})/m}}K!^{K-\floor{\ell(\boldsymbol{c})/m}}.
    \end{equation}
\end{proposition}

\begin{proof}
    Fix indices $j,k$ such that $c_j\ge c_k\ge 1$ and define $\tilde{\boldsymbol{c}}\in[0,m]^K$ by $\tilde{c_j} = c_j + 1$,  $\tilde{c_k} = c_k - 1$, and $\tilde{c_i} = c_i$ for all $i\notin\{j,k\}$.
    We have that
    \begin{equation}
        \prod_{i=1}^K \frac{\tilde{c}_i! (K-\tilde{c}_i)!}{c_i! (K-c_i)!} = \frac{(c_j + 1)(K - (c_k - 1))}{c_k(K - c_j)} > 1.
    \end{equation}
    Hence, moving weight from a column to one with equal or larger weight increases the desired value. 
    If each column has an upper bound of $m$ and a lower bound of 0, then this weight-shifting device leads to an optimized $\tilde{c}$ made up of 0's, $\floor{\ell(\boldsymbol{c})/m}$ $m$'s and possibly one column in $(0,m)$. 
    Since $(K-c)!c! \leq K!$ for all $c\in[K]$, the ``remainder'' column can be bounded by $K!$. The result follows.
\end{proof}

\begin{proposition}\label{co:comb}
    Let $j\in[K]$. For fixed $\boldsymbol{c}\in{[0,K]}^K$ with $\ell(\boldsymbol{c}) < K$ and exactly $j$ nonzero elements,
    \begin{equation}
    \prod_{i=1}^K c_i! (K-c_i)! \leq (\ell(\boldsymbol{c}) - j + 1)!(K - \ell(\boldsymbol{c}) + j - 1)(K-1)!^{j-1}K!^{\ell(\boldsymbol{c})-j}.
    \end{equation}
\end{proposition}

\begin{proof}
    As with \cref{pr:comb}, we know that the maximal value for this product is obtained when we can no longer shift weight from a smaller column to a larger column. If $\ell(\boldsymbol{c}) < K$ and $j$ elements of $\boldsymbol{c}$ are nonzero, then this arrangement is obtained when $j-1$ elements of $\boldsymbol{c}$ are 1 and there is a unique element equal to $\ell(\boldsymbol{c}) - (j-1)$.
\end{proof}

\begin{proposition}\label{pr:prod-ratio}
    Fix $m\in[K]$. For any $\boldsymbol{c}\in[0,m]^K$ with $\ell(\boldsymbol{c}) < K$,
    \begin{equation}
        \prod_{i=1}^K \frac{\binom{m}{c_i}}{\binom{K}{c_i}} \leq \left(\frac{m}{K}\right)^{\ell(\boldsymbol{c})}.
    \end{equation}
\end{proposition}

\begin{proof}
    Fix indices $j, k$ such that $c_j \geq c_k\geq 1$, and define $\tilde{\boldsymbol{c}}\in[0,m]^K$ as in the proof of \cref{pr:comb}. We have that
    \begin{equation}
        \begin{aligned}
            \prod_{i=1}^K \frac{\binom{m}{\tilde{c}_i}\binom{K}{c_i}}{\binom{K}{\tilde{c}_i}\binom{m}{c_i}} &= \frac{(K - \tilde{c}_j)!(K - \tilde{c}_k)!(m - c_j)!(m - c_k)!}{(m - \tilde{c}_j)!(m - \tilde{c}_k)!(K - c_j)!(K - c_k)!}
            = \frac{(m - c_j)(K - c_k + 1)}{(m - c_k + 1)(K - c_j)}.
        \end{aligned}
    \end{equation}
    As the function $x/(x+n)$ increases in $x$ for any $n>0$, the ratio above is less than 1. 
    Hence, moving weight from one column to another with equal or larger weight decreases the product in question, and the maximum value is obtained when all columns have equal weighting. As $\ell(\boldsymbol{c}) < K$, this means that $\ell(\boldsymbol{c})$ elements of $\boldsymbol{c}$ have a value of 1, and all other elements are 0.
\end{proof}

\section{Better response dynamics: proof of \cref{th:main1} (2)}

While there are clear-cut conditions that indicate when a BRD process has entered a trap, that luxury unfortunately does not extend to the domain of bRD processes. For $\mathsf{bRD}$ to determine that it has entered a trap based solely on its past, it must exhaust all movement options from every strategy profile in the suspected trap. This definition of a trap is unwieldy at best, so we require a different approach to address the question of convergence for $\mathsf{bRD}$.

For a $\mathsf{bRD}$-trap $\tau$, denote by $\boldsymbol{R}(\tau)$ (resp. $\boldsymbol{C}(\tau)$) the $K$-dimensional vector whose $j$-th entry is the number of strategy profiles in $\tau$ for which the first player (resp. second player) chooses strategy $j$.
The length of a trap is the number of strategy profiles that it contains. 
We denote by $\mathcal{T}_n$ the collection of $\mathsf{bRD}$-traps of length $n$; further, let $\mathcal{T} = \bigcup_{n=4}^{K^2} \mathcal{T}_n$.
For strategy $i\in [K]$, let 
\begin{equation*}
    M_{i,u} :=\{k\colon Z^{i, k}_2 \text{ is among the largest $u$ payoffs in row $i$}\},
\end{equation*}
and define $\Delta^{\sigma}_1 \coloneqq  \bigcup_{i,j} \bigcap_{k\in M_{i,K^\sigma}}\{Z^{i,k}_2 > Z^{j,k}_2\}$.

\begin{lemma}
\label{le:P(C1)}
Fix a parameter $\sigma\in(0,1)$. We have
$\mathbb{P}(\Delta^{\sigma}_1) \le K^{2} \parens*{\frac{1}{2}}^{K^{\sigma}}.$
\end{lemma}
\begin{proof}
Using a union bound it is enough to compare two fixed distinct rows, say $i$ and $j$, and this gives rise to the $K^2$ factor. For the remaining part of the upper bound,  notice that for fixed $k\in [K]$, the events $\{Z^{i,k}_2 > Z^{j,k}_2\}$, for $k \in M_{i,K^\sigma}$, are independent and share the same probability of $1/2$. 
\end{proof}

The nonexistence of traps implies that $\mathsf{bRD}$ will converge. Letting $J_K := \braces{\PNE_K\neq\varnothing}$,
\begin{equation*}
    \mathbb{P}(\mathsf{bRD}\text{ does not converge}) \leq \mathbb{P}\parens*{\braces*{\mathcal{T}\neq\varnothing} \cap J_K} + \mathbb{P}\parens*{{J_K}^c}.
\end{equation*}
Owing to the previously mentioned Poisson result from \cite{Pow:IJGT1990}, to prove \cref{th:main1} 2), it suffices to prove the following theorem.

\begin{theorem}
    \label{th:trap-PNE}
    For any $\alpha \in (0,1)$, we have  that there exists a positive constant, denoted by $\cnst$\hspace{-.1cm}, such that for all large enough $K$,
    \begin{equation}
        \label{eq:trap-PNE}
        \mathbb{P}\parens*{\braces*{\mathcal{T}\neq\varnothing} \cap J_K}
        \le  \cnst K^{-3+2\alpha} + K^{2} \parens*{\frac{1}{2}}^{K^{\alpha/2}} +   K^{-K^\alpha }.
    \end{equation}
\end{theorem}

\begin{proof}
In what follows, $cnst$ denotes a generic constant that may change from line to line. 
    We have that
    \begin{equation}
        \label{eq:t-pne}
        \mathbb{P}\parens*{\braces*{\bigcup_{n=4}^{K^2} \mathcal{T}_n \neq\varnothing} \cap J_K}
        \leq \mathbb{P}\parens*{\bigcup_{n=4}^{\floor{K^\alpha}}\mathcal{T}_n \neq\varnothing} + \mathbb{P}\parens*{\braces*{\bigcup_{n=\floor{K^\alpha} + 1}^{K^2} \mathcal{T}_n \neq\varnothing} \cap J_K}.
    \end{equation}
    For a given subset of strategy profiles $\mathcal{V}$, let $\boldsymbol{\Gamma}(\mathcal{V})$ be the projection of $\mathcal{V}$ on the second coordinate. Observe that for $\mathcal{V}$ to be a trap, in every row and column intersecting with $\mathcal{V}$, there must not exist strategy profiles that are better responses than those in $\mathcal{V}$. If there are $v$ vertices of $\mathcal{V}$ in a given row or column, the probability that these $v$ vertices attain the highest $v$ payoffs is $\binom{K}{v}$. Moreover, this event is independent from the payoffs in every other row and column. Hence, if $\mathcal{V}$ has length $n$, then
    \begin{equation*}
        \mathbb{P}(\mathcal{V}\in\mathcal{T}_n) \leq \left(\prod_{i=1}^K \binom{K}{R_i(\mathcal{V})}\binom{K}{C_i(\mathcal{V})}\right)^{-1}. 
    \end{equation*}
    The inequality arises from the fact that certain orderings may not produce a trap, e.g., a row and column may have their best responses on the same strategy profile, resulting in a PNE.
    
    To bound the first term in \cref{eq:t-pne}, we apply a union bound:
    \begin{equation}
        \begin{aligned}
            \mathbb{P}\parens*{\bigcup_{n=4}^{\floor{K^\alpha}}\mathcal{T}_n \neq\varnothing} &\leq \sum_{n=4}^{\floor{K^\alpha}}\sum_{j=2}^{n-2} \mathbb{P}\parens*{\exists \mathcal{V}\subset[K]^2: \mathcal{V}\in\mathcal{T}_n, \card{\boldsymbol{\Gamma}(\mathcal{V})} = j}\\
            &\leq \sum_{n=4}^{\floor{K^\alpha}}\sum_{j=2}^{n-2} \binom{K}{j}\sum_{\substack{\boldsymbol{r}\in{[0,K]}^K:\\\ell(\boldsymbol{r})=n,\max(r_i)\leq j}} \sum_{\substack{\mathcal{V}\subset[j]^K:\\\boldsymbol{R}(\mathcal{V}) = \boldsymbol{r},\boldsymbol{\Gamma}(\mathcal{V})=[j]}}\frac{1}{\prod_{i=1}^K\binom{K}{R_i(\mathcal{V})}\binom{K}{C_i(\mathcal{V})}}\\
            &\leq \sum_{n=4}^{\floor{K^\alpha}}\sum_{j=2}^{n-2} \binom{K}{j}\sum_{\substack{\boldsymbol{r}\in{[0,K]}^K:\\\ell(\boldsymbol{r})=n,\max(r_i)\leq j}} \frac{1}{\prod_{i=1}^K\binom{K}{r_i}} \sum_{\substack{\mathcal{V}\subset[j]^K:\\\boldsymbol{R}(\mathcal{V}) = \boldsymbol{r},\boldsymbol{\Gamma}(\mathcal{V})=[j]}}\frac{1}{\prod_{i=1}^j\binom{K}{C_i(\mathcal{V})}}.
        \end{aligned}
    \end{equation}
    We use \cref{co:comb} to obtain the upper bound
    \begin{equation}
        \prod_{i=1}^j (K-C_i(\mathcal{V}))!C_i(\mathcal{V})! \leq (\ell(\boldsymbol{C}(\mathcal{V})) - j + 1)!(K - \ell(\boldsymbol{C}(\mathcal{V})) + j - 1)!(K-1)!^{j-1}.
    \end{equation}
    Moreover, for a given vector $\boldsymbol{r}$ with each element being at most $j$, there are $\prod_{i=1}^K \binom{j}{r_i}$ sets of strategy profiles $\mathcal{V}$ satisfying $\boldsymbol{R}(\mathcal{V}) = \boldsymbol{r}$. This gives
    \begin{equation}
        \mathbb{P}\parens*{\bigcup_{n=1}^{\floor{K^\alpha}}\mathcal{T}_n \neq\varnothing} \leq \sum_{n=4}^{\floor{K^\alpha}}\sum_{j=2}^{n-2} \frac{(K-n+j-1)!(n-j+1)!}{K^{j-1}j!(K-j)!}\sum_{\substack{\boldsymbol{r}\in{[0,K]}^K:\\\ell(\boldsymbol{r})=n,\max(r_i)\leq j}} \prod_{i=1}^K\frac{\binom{j}{r_i}}{\binom{K}{r_i}}.
    \end{equation}
    \cref{pr:prod-ratio} allows us to bound the product above, and a loose bound for the number of vectors $\boldsymbol{r}$ with $\ell(\boldsymbol{r}) = n$ is $\binom{n+K-1}{K}$. Hence,
    \begin{equation}\label{eqNoSmallTraps}
    \begin{aligned}
        \mathbb{P}\parens*{\bigcup_{n=1}^{\floor{K^\alpha}}\mathcal{T}_n \neq\varnothing} &\leq \sum_{n=4}^{\floor{K^\alpha}}\sum_{j=2}^{n-2} \frac{(K-n+j-1)!(n-j+1)!}{K^{j-1}j!(K-j)!}\binom{n+K-1}{K}\left(\frac{j}{K}\right)^n \le  \cnst K^{-3+2\alpha},
        \end{aligned}
    \end{equation}
    where the last inequality is proved in \cref{pr:Kais1} in the Appendix.

    To bound the remaining term in \cref{eq:t-pne}, we make use of the fact that for a trap $\tau$ to be larger than $K^\alpha$, either: there must exist a row $r$ or a column $c$ containing at least $K^{\alpha/2}$ strategy profiles of $\tau$; or no such row/column exists and the trap spans more than $K^{\alpha/2}$ rows and columns.
    These events will be denoted as $A_1$ and $A_2$ respectively.
    Note that any strategy profile neighbouring a PNE cannot be part of a trap; hence, if $\boldsymbol{s}$ and $\boldsymbol{t}$ are such that $\boldsymbol{s}\sim_i \boldsymbol{t}$, $\boldsymbol{s}$ neighbours a PNE, and $\boldsymbol{t}$ is in a trap, then $Z_i^{\boldsymbol{s}} < Z_i^{\boldsymbol{t}}$.
    Under event $A_1$, we can find at least $K^{\alpha/2}$ such pairs of strategy profiles: $\boldsymbol{t}$ chosen from row $r$ and $\boldsymbol{s}$ chosen from any row containing a PNE.
    Hence, we are in the framework of \cref{le:P(C1)}, and by a union bound on rows and columns we have
    \begin{equation}
        \mathbb{P}\parens*{\braces*{\bigcup_{n=\floor{K^\alpha} + 1}^{K^2} \mathcal{T}_n \neq\varnothing} \cap J_K \cap A_1} \leq 2K^{2} \parens*{\frac{1}{2}}^{K^{\alpha/2}}.
    \end{equation}
    
    We turn our attention to event $A_2$.  By taking a union bound over all trap sizes larger than $K^\alpha$, and applying \cref{pr:comb}, we get
    \begin{multline}\label{eq:bb1}
        \mathbb{P}\parens*{\braces*{\bigcup_{n=\floor{K^\alpha} + 1}^{K^2} \mathcal{T}_n \neq\varnothing} \cap J_K \cap A_2}
        \leq \sum_{n=\floor{K^\alpha}}^{K^2} \sum_{\mathcal{V}: \card{\mathcal{V}} = n} \frac{1}{(K!)^K \prod_{i=1}^K \binom{K}{R_i(\mathcal{V})}}\prod_{i=1}^K c_i!(K - c_i)!\\
        \leq \sum_{n=\floor{K^\alpha}}^{K^2} \sum_{\mathcal{V}: \card{\mathcal{V}} = n} \frac{(\floor{K^{\frac{\alpha}{2}}}!(K-\floor{K^{\frac{\alpha}{2}}})!)^{\floor{n /\floor{K^{\alpha/2}}}}K!^{K-\floor{n /\floor{K^{\alpha/2}}}}}{(K!)^K \prod_{i=1}^K \binom{K}{R_i(\mathcal{V})}}.
    \end{multline}
    For any $\boldsymbol{r}\in{[0,K]}^K$ with $\ell(\boldsymbol{r}) = n$, there are $\prod_{i=1}^K \binom{K}{r_i}$ strategy profile sets $\mathcal{V}$ of size $n$ which satisfy $\boldsymbol{R}(\mathcal{V}) = \boldsymbol{r}$; moreover, there are $\binom{n+K-1}{K}$ such vectors. Hence the right hand side of \cref{eq:bb1} becomes
    \begin{equation}\label{eqNoCoexistenceTrapsPNE}
        \sum_{n=\floor{K^\alpha}}^{K^2} \binom{n+K-1}{K}\binom{K}{\floor{K^{\alpha/2}}}^{-\floor{n /\floor{K^{\alpha/2}}}} < \cnst K^{-\beta K^\alpha}
    \end{equation}
    where the last inequality  is proved in \cref{pr:Kais2} in the Appendix.
\end{proof}

\section{Appendix}
Here we prove the estimates for Equations \eqref{eqNoSmallTraps} and \eqref{eqNoCoexistenceTrapsPNE}. We use the following approximations 
$$e^{1/(12n+1)} < \frac{n!e^n}{\sqrt{2\pi}n^{n+1/2}} < e^{1/(12n)}\text{ and }e^{a+a^2/(2x)-a^3/(2x^2)}<(1+a/x)^{x+a}<e^{a+a^2/x}$$
We let $\phi(x)=x^x=e^{x\ln x}$ ($\phi(0)=1$), $A(a,x)=e^{a+a^2/x}$ and $B(a,x)=e^{a+a^2/(2x)-a^3/(2x^2)}$.
\begin{proposition}\label{pr:Kais1} For any $\alpha \in (0,1)$ there exists a $cnst$ that satisfies
$$\sum_{n=4}^{\floor{K^\alpha}}\sum_{j=2}^{n-2} \frac{(K-n+j-1)!(n-j+1)!}{K^{j-1}j!(K-j)!}\binom{n+K-1}{K}\left(\frac{j}{K}\right)^n\le \cnst K^{-3+2\alpha},$$
for all $K \in \N$.
\end{proposition}
\begin{proof}
Let
$\begin{displaystyle}G(K,n,j) = \frac{\phi(K-n-1+j)\phi(n+1-j)}{\phi(j)\phi(K-j)K^{j-1}}\frac{\phi(K+n-1)}{\phi(K)\phi(n-1)}\left(\frac{j}{K}\right)^n\end{displaystyle}$.
An examination of the behaviour of the function $G$ leads to 
$G(K,n,j) \le G(K,j+2,j) \le \cnst K^{-3}$, for $K$ large enough. For the second inequality, it is necessary to distinguish the cases $j=2$, $j=3$ and $j\ge4$.

A bound for the square root term (that which arises from Stirling's approximation) can be obtained in a similar fashion, from which we immediately get the desired result.
\end{proof}

\begin{proposition}\label{pr:Kais2} For any $\alpha \in (0,1)$ and for any $\beta<\alpha/2$, we have that for any $K$ large enough (depending on the choice of $\alpha$ and $\beta$), 
\begin{equation*}
    \sum_{n=\floor{K^\alpha}}^{K^2} \binom{n+K-1}{K}\binom{K}{\floor{K^{\alpha/2}}}^{-\floor{n /\floor{K^{\alpha/2}}}} < K^{-\beta K^\alpha}.
\end{equation*}
\end{proposition}
\begin{proof}
For $K$ large, $\begin{displaystyle}\psi(K) = \left(\frac{K!}{N!(K-N)!}\right)^{1/N} > (K-N)/(2N)\end{displaystyle}$, where for ease of presentation, we write $N$ for $\floor{K^{\alpha/2}}$. Then,
$$G(K,n)=\frac{\phi(K+n-1)}{\phi(K)\phi(n-1)}\psi(K)^{-n}$$
is a decreasing function of $n$ and $G(K,n) \le  G(K,K^\alpha-1) < \cnst e^{3K^\alpha/2+K^{-1+2\alpha}-1/2}K^{-(\alpha/2)K^{\alpha}-1+3\alpha/2}$.
Similarly,
$$\sqrt{\frac{K}{N(K-N)}}\le \cnst K^{-\alpha/4}\text{ and }
\frac{\phi(K)}{\phi(N)\phi(K-N)} < K^{(1-\alpha/2)K^{\alpha/2}}e^{K^{\alpha/2}-\frac12K^{-1+\alpha}-\frac12K^{-2+3\alpha/2}}.$$
The result immediately follows.
\end{proof}


\bibliography{bibNEpercolation}
\bibliographystyle{apalike}

\end{document}